\documentclass[a4paper,11pt]{article}
\pdfoutput=1
\usepackage{microtype}
\usepackage[a4paper,margin=1.2in]{geometry}
\usepackage[UKenglish]{babel}
\usepackage[numbers,sort&compress]{natbib}
\usepackage{color}
\usepackage{doi}
\usepackage{latexsym}
\usepackage{amsmath}
\usepackage{amssymb}
\usepackage{amsthm}
\usepackage{euscript}
\usepackage[all,warning]{onlyamsmath}
\usepackage{enumitem}
\usepackage{xparse}
\usepackage{mathtools}
\usepackage{graphicx}
\usepackage[small]{caption}
\usepackage{wrapfig}
\usepackage{hyperref}

\theoremstyle{plain}
\newtheorem{theorem}{Theorem}
\newtheorem{lemma}[theorem]{Lemma}

\newtheorem{corollary}[theorem]{Corollary}

\theoremstyle{definition}
\newtheorem{definition}[theorem]{Definition}

\theoremstyle{remark}

\setlist[itemize]{label=--}
\setlist[enumerate]{label=(\arabic*),labelindent=\parindent,leftmargin=*}

\DeclareMathOperator{\indeg}{indeg}
\DeclareMathOperator{\outdeg}{outdeg}

\DeclareMathOperator{\poly}{poly}

\DeclareMathOperator{\ex}{ex}

\newcommand{\N}{\mathbb{N}}

\newcommand{\A}{\EuScript{A}}
\newcommand{\B}{\EuScript{B}}
\newcommand{\paths}{\EuScript{P}}
\newcommand{\trees}{\EuScript{T}}
\newcommand{\nout}{N_\text{out}}
\newcommand{\nin}{N_\text{in}}

\newcommand{\local}{\ensuremath{\mathsf{LOCAL}}}
\newcommand{\congest}{\ensuremath{\mathsf{CONGEST}}}
\newcommand{\pram}{\ensuremath{\mathsf{PRAM}}}

\newcommand{\namedref}[2]{\hyperref[#2]{#1~\ref*{#2}}}
\newcommand{\sectionref}[1]{\namedref{Section}{#1}}

\newcommand{\theoremref}[1]{\namedref{Theorem}{#1}}

\newcommand{\figureref}[1]{\namedref{Figure}{#1}}
\newcommand{\lemmaref}[1]{\namedref{Lemma}{#1}}

\newenvironment{myabstract}
               {\list{}{\listparindent 1.5em%
                        \itemindent    \listparindent
                        \leftmargin    1cm
                        \rightmargin   1cm
                        \parsep        0pt}%
                \item\relax}
               {\endlist}

\newenvironment{mycover}
               {\list{}{\listparindent 0pt
                        \itemindent    \listparindent
                        \leftmargin    1cm
                        \rightmargin   1cm
                        \parsep        0pt}%
                \raggedright
                \item\relax}
               {\endlist}

\newcommand{\myemail}[1]{\,$\cdot$\, {\small #1}\par\vspace{2pt}}
\newcommand{\myaff}[1]{{\small #1\par}\bigskip}

\renewcommand{\poly}{\operatorname{poly}}

\hypersetup{
  colorlinks=true,
  linkcolor=black,
  citecolor=black,
  filecolor=black,
  urlcolor=[rgb]{0,0.1,0.5},
  pdftitle={Deterministic subgraph detection in broadcast CONGEST}, 
  pdfauthor={Janne H. Korhonen and Joel Rybicki}
}

\begin{document}

\mbox{}
\begin{mycover}
{\huge \bfseries Deterministic subgraph detection in broadcast CONGEST \par}
\bigskip
\bigskip

\textbf{Janne H.\ Korhonen}
\myemail{janne.h.korhonen@aalto.fi}
\myaff{Aalto University}

\textbf{Joel Rybicki}
\myemail{joel.rybicki@helsinki.fi}
\myaff{University of Helsinki}

\end{mycover}

\bigskip
\begin{myabstract}
\noindent\textbf{Abstract.} We present simple deterministic algorithms for subgraph finding and enumeration in the broadcast \congest{} model of distributed computation:
\begin{itemize}
    \item For any constant $k$, detecting $k$-paths and trees on $k$ nodes can be done in $O(1)$ rounds.
    \item For any constant $k$, detecting $k$-cycles and pseudotrees on $k$ nodes can be done in $O(n)$ rounds.
    \item On $d$-degenerate graphs, cliques and $4$-cycles can be enumerated in $O(d + \log n)$ rounds, and $5$-cycles in $O(d^2 + \log n)$ rounds.
\end{itemize}
In many cases, these bounds are tight up to logarithmic factors. Moreover, we show that the algorithms for $d$-degenerate graphs can be improved to optimal complexity $O(d/\log n)$ and $O(d^2/\log n)$, respectively, in the supported \congest{} model, which can be seen as an intermediate model between \congest{} and the congested clique.
\end{myabstract}
\thispagestyle{empty}
\setcounter{page}{0}
\newpage

\section{Introduction}

\emph{Subgraph detection} is a fundamental problem in algorithmics: given a fixed target graph $H$ and an input graph $G$, the task is to decide whether $G$ contains a subgraph isomorphic to $H$. This problem has been extensively studied in centralised algorithmics, and it is known that depending on the graph $H$, this problem can have very different complexities---for example, if $H$ is a path, the problem can be solved in linear (in the number of nodes in $G$) time~\cite{alon1995color}, but for cliques this is thought not to be the case~\cite{Downey94parameterizedcomputational}.

In this work, we study this problem in the \congest{} model of distributed computation. So far, most  work has either focused on (a)~lower bounds for subgraph detection~\cite{drucker13}, (b)~randomised upper bounds~\cite{izumi2017,firscher2017}, or (c)~property testing of $H$-freeness~\cite{Fraigniaud2016,Censor-Hillel2016,firscher2017}. By contrast, we focus on deterministic upper bounds, showing that simple techniques---indeed, often techniques well-known in non-distributed algorithmics---result in essentially optimal algorithms. Moreover, all our results work also in the weaker \emph{broadcast \congest{}} model, where nodes send the same message to all neighbours in each communication round.

\paragraph{Results: subgraph detection on general graphs.}

First, we give simple constant-time detection algorithms for the case where $H$ is either a path or a tree. These algorithms are essentially translations of known centralised \emph{fixed-parameter} algorithms, based on \emph{representative families}~\cite{marx2009parameterized,fomin2014efficient,monien1985find}. Specifically, we show that in the broadcast \congest{} model,
\begin{itemize}
 \item $k$-paths can be detected in $O(k2^k)$ rounds, and
 \item any tree on $k$ nodes can be detected in $O(k2^k)$ rounds\footnote{We note that the constant-round algorithms for path and tree detection were also independently discovered by Fraigniaud et al.~\cite{fraigniaud2017arxiv}; however, we give more detailed analysis in terms of the factors dependent on $k$. In fact, very similar algorithms were independently discovered in \emph{four} separate works almost concurrently, including this one~\cite{this-arxiv} and the papers of Fraigniaud et al.~\cite{fraigniaud2017arxiv}, Even et al.~\cite{2017arXiv170504898E} and Fischer et al.~\cite{firscher2017}. The last three appear as a joint paper in DISC~2017~\cite{disc2017property}. While our work is otherwise independent of the others, our pseudotree detection algorithm is motivated by the pseudotree property testing results of Fraigniaud et al.~\cite{fraigniaud2017arxiv}.}. 
\end{itemize}
Using these algorithms as a building block, we then give linear-time detection algorithms for cycles and pseudotrees:
\begin{itemize}
 \item $k$-cycles can be detected in $O(k2^kn)$ rounds, and
 \item any pseudotree (a graph with exactly one cycle) on $k$ nodes can be detected in $O(k2^kn)$ rounds.
\end{itemize}
All algorithms can be implemented so that any node that detects the existence of a copy of $H$ can also output full information about a single copy of $H$.

For odd $k \ge 5$, Drucker et al.~\cite{drucker13} have proven a lower bound of $\Omega(n/\log n)$ rounds for detecting $k$-cycles, meaning that our cycle detection algorithm is optimal up to a logarithmic factor for a constant odd $k$. For $4$-cycles, Drucker et al.~\cite{drucker13} gave a lower bound of $\Omega(n^{1/2}/\log n)$ rounds, and a weaker lower bound for longer even cycles (see \sectionref{sec:related-work}); we prove a simple extension of the Drucker et al.~\cite{drucker13} lower bound by showing that for all even $k \ge 6$, detecting $k$-cycles requires $\Omega(n^{1/2}/\log n)$ rounds in the \congest{} model.

\paragraph{Results: subgraph enumeration on sparse graphs.}
Second, we study subgraph detection and enumeration on sparse input graphs. Specifically, we study a setting where the input graph has small \emph{degeneracy}; in distributed and parallel computing, this setting has been studied in the context of e.g.\ symmetry breaking in the \local{} model~\cite{Barenboim2010,barenboim16locality} and enumerating triangles and 4-cliques of planar graphs in the \pram{} model~\cite{chrobak91orientations}. Denoting by $d$ the degeneracy of the input graph, we show that in the broadcast \congest{} model,
\begin{itemize}
    \item $k$-cliques can be enumerated in $O(d + \log n)$ rounds for any $k$,
    \item $3$-cycles and $4$-cycles can be enumerated in $O(d + \log n)$ rounds, and
    \item $5$-cycles can be enumerated in $O(d^2 + \log n)$ rounds,
\end{itemize}
where enumeration means that every copy of the target subgraph $H$ is output by some node in the network.

This is as far as we can push this framework; we show that already detecting $k$-cycles for $k \ge 6$ requires $\Omega(n^{1/2}/\log n)$ rounds on graphs of degeneracy $2$. Moreover, a careful examination of the Drucker et al.~\cite{drucker13} lower bound constructions show that detecting $4$-cycles and $5$-cycles requires $\Omega(d/\log n)$ rounds.
 
Finally, we discuss how the results on detecting subgraph in sparse graphs can be translated into the \emph{supported \congest{} model} proposed by Schmid and Suomela~\cite{schmid13local-sdn}. In this model, we can close many of the logarithmic gaps between the upper and lower bounds that remain in the \congest{} model. 

\section{Related work}\label{sec:related-work}

\paragraph{Subgraph detection upper bounds.} For deterministic subgraph detection in the \congest{} model, the only prior works we are aware of are the $O(n^{1/2})$ round algorithm for $4$-cycle detection by~Drucker et al.~\cite{drucker13}, and the independent discovery of the constant-round path and tree detection algorithms~\cite{fraigniaud2017arxiv,disc2017property}. In the \emph{congested clique} model, deterministic subgraph detection algorithms were given by Dolev et al.~\cite{tritri} and Censor-Hillel et al.~\cite{censor2015algebraic}; the latter is in particular noteworthy from our perspective, as it used \emph{colour-coding} techniques from centralised fixed-parameter algorithmics to obtain fast cycle detection algorithms for cycles of arbitrary length.

Randomised subgraph detection and listing in the \congest{} model has been recently receiving attention from multiple authors. Izumi and Le Gall~\cite{izumi2017} gave an $\tilde{O}(n^{2/3})$ round algorithm for detecting triangles, and a $O(n^{3/4}\log n)$ round algorithm for enumerating triangles. Independent of our work, Fischer et al.~\cite{firscher2017} gave a colour-coding algorithm that can detect any constant-size tree on $k$ nodes with constant probability in $O(k^k)$ rounds.

\paragraph{Subgraph detection lower bounds.} As noted before, Drucker et al.~\cite{drucker13} have studied lower bounds for cycle detection in the \congest{} model. Specifically, their lower bound for $k$-cycle detection is $\Omega(\ex(n,C_k)/n \log n)$ rounds, where $\ex(n,C_k)$ is the \emph{Tur\'{a}n number} for cycles, that is, the maximum number of edges in a $k$-cycle-free graph with $n$ nodes. This lower bound is $\Omega(n/\log n)$ for odd cycles, and $\Omega(n^{1/2}/\log n)$ for $4$-cycles. However, for longer even cycles, this can give at most $\Omega(n^{1+2/k}/\log n)$ due to known bounds for Tur\'{a}n numbers, and matching bounds on Tur\'{a}n numbers are only known for $k = 6 $ and $k = 8$; see e.g.\ Pikhurko~\cite{pikhurko2012} and references therein.

To our knowledge, no other subgraph detection lower bounds are known in the \congest{} model. In particular, proving lower bounds for triangle detection seems to be a particularly difficult challenge. However, for the broadcast congested clique model, Drucker et al.~\cite{drucker13} give an $\Omega(n/e^{\sqrt{\log n} }\log n)$ round lower bound, which also applies to the broadcast \congest{} model. Moreover, for triangle enumeration lower bounds are known, also in the stronger congested clique model~\cite{izumi2017,pandurangan2016tight}.

\paragraph{Property testing for $H$-freeness.} Property testing of $H$-freeness in the \congest{} model is another question that has received a lot of attention lately. In this setting, an algorithm has to correctly decide with probably with probability at least $2/3$ if the input graph is (a) $H$-free---that is, does not contain a subgraph isomorphic to $H$---or (b) $\varepsilon$-away from being $H$-free; in the intermediate case, the algorithm can perform arbitrarily. See e.g. Censor-Hillel et al.~\cite{Censor-Hillel2016} for complete definitions.

Property testing algorithms for triangle-freeness were given by Censor-Hillel et al.~\cite{Censor-Hillel2016}; for $H$-freeness for graphs on $4$ nodes by Fraigniaud et al.~\cite{fraigniaud2017arxiv} and Even et al.~\cite{2017arXiv170504898E}; and for $H$-freeness for most graphs on $5$ nodes by Fischer et al.~\cite{firscher2017} Moreover, property testing algorithms for tree and cycle freeness---using techniques of fixed-parameter algorithmics flavour---have been discovered recently~\cite{2017arXiv170504898E,firscher2017,fraigniaud2017arxiv,disc2017property}.

\section{Preliminaries}

\paragraph{Set notation.}
Let $\N = \{0,1,\ldots\}$ be the set of non-negative integers. For integer $n \in \N$, we use the notation $[n] = \{ 1,2,\dotsc,n \}$. For any set $A$, we use $2^A = \{ B : B \subseteq A \}$ to denote the power set of $A$.

\paragraph{Graphs.} A  graph is a pair $G = (V,E)$, where $V$ is the set of nodes and $E \subseteq 2^V$ is the set of edges. We use $n = |V|$ to denote the number of nodes in the graph. An \emph{orientation} $\sigma$ of graph $G$ is a labelling of the edges that assigns a direction $\sigma(\{u,v\}) \in \{ u \rightarrow v, u \leftarrow v\}$ to every edge $\{u,v\} \in E$.

The open neighbourhood of a node $v \in V$ is the set $N(v) = \{ u \in V : \{u,v\} \in E\}$ and the closed neighbourhood is the set $N^+(v) = N(v) \cup \{ v \}$. The degree of a node $v \in V$ is $\deg(v) = |N(v)|$. The neighbours of $v \in V$ with incoming and outgoing edges under an orientation $\sigma$ are denoted by $\nin(v) = \{ u \in N(v) : \sigma(\{u,v\}) = u \rightarrow v\}$ and $\nout(v) = N(v) \setminus \nin(v)$. The indegree of $v \in V$ under an orientation $\sigma$ is $\indeg(v) = |\nin(v)| $ and the outdegree is given by $\outdeg(v) = |\nout(v)|$. 

For a graph $G$, a subgraph $G' \subseteq G$ of $G$ is a graph $G' = (V',E')$, where $V'\subseteq V$ and $E' \subseteq E \cap 2^{V'}$. The subgraph induced by the node set $A \subseteq V$ is $G[A] = (A, E')$, where $E' = \{ e \in E : e \subseteq A \}$. Similarly, the subgraph induced by an edge set $F \subseteq E$ is $G[F] = (V',F)$, where $V' = \{ u \in e : e \in F\}$. A graph $G$ is \emph{$d$-degenerate} if every subgraph $G' \subseteq G$ contains a node with degree at most $d$. Note that every graph is trivially $(n-1)$-degenerate and any graph with arboricity $a$ has degeneracy $a \le d \le 2a-1$. 

\paragraph{\boldmath \congest{} and broadcast \congest{}.} The \congest{} model is a variant of classical \local{} model of distributed computation with additional constraints on communication bandwidth~\cite{peleg00distributed}. The distributed system is represented as a network $G = (V,E)$, where each node $v \in V$ executes the same algorithm in synchronous rounds, and the nodes collaborate to solve a graph problem with input $G$. Each round, all nodes
\begin{enumerate}
\item perform an unlimited amount of local computation,
\item send a possibly different $O(\log n)$-bit message to each of their neighbours, and
\item receive the messages sent to them.
\end{enumerate}
The time measure is the number of synchronous rounds required. 
Each node is assumed to have a unique identifiers from the set $\{0, \ldots, \poly(n)\}$.

The broadcast \congest{} is a weaker version of \congest{}, with the additional constraint that all nodes have to send the same message to each of their neighbours.

\section{Finding paths, cycles and trees}

\paragraph{Representative families.}
The general cycle and path detection algorithms are based on \emph{representative families}~\cite{monien1985find,erdos-representative}. For a basic intuition, consider a setting where we have a collection of objects (in this case, sets) of size $p$, and we want to extend them by adding at most $q$ more elements. Representative families allow us to `compress' our collection of objects so that if a member of the original collection could be extended by specific $q$ elements, then the compressed collection also contains a member that can be extended by the same elements. Indeed, while we only use the set family version of this theory, it can be extended to \emph{matroids}~\cite{fomin2014efficient}.

\begin{definition}
Let $\A \subseteq 2^{[n]}$. We say that $\widehat{\A} \subseteq \A$ is \emph{$q$-representative} for $\A$ if for each $B \subseteq [n]$ with $|B| \le q$, there is a set $A \in \A$ with $A \cap B = \emptyset$ if and only if there is $A \in \widehat{\A}$ with $A \cap B = \emptyset$.
\end{definition}


\begin{theorem}[\cite{erdos-representative}]\label{thm:rep-size}
Let $\A \subseteq 2^{[n]}$ consist of sets of size at most $p$. Then there is a $q$-representative $\widehat{\A} \subseteq \A$ with $|\widehat{\A}| \le \binom{p + q}{p}$.
\end{theorem}

For our purposes it is sufficient to know that small representative families exists; however, slightly worse bounds with corresponding efficient algorithms are known~\cite{marx2009parameterized,fomin2014efficient,erdos-representative}.

Finally, we will use the following simple fact about representative families:

\begin{lemma}\label{lemma:rep-transitivity}
Let $\EuScript{C} \subseteq \B \subseteq \A \subseteq 2^{[n]}$. If $\B$ is $q$-representative for $\A$ and $\EuScript{C}$ is $q$-representative for $\B$, then $\EuScript{C}$ is $q$-representative for $\A$.
\end{lemma}

Together \theoremref{thm:rep-size} and \lemmaref{lemma:rep-transitivity} imply that any inclusion-minimal $q$-representative family for a family of sets of size at most $p$ has size at most $\binom{p + q}{p}$.

\paragraph{Finding paths.}

We start by giving a simple path detection algorithm for the \congest{} model, using a well-known technique from fixed-parameter algorithmics~\cite{monien1985find}. Intuitively, the idea is to start with the trivial path detection algorithm that iteratively propagates full information about paths of length at most $\ell = 1, 2, \dotsc, k$ in $k$ phases. This will be quite slow if there are lots of such paths; however, at each step, we construct a representative family for the current set of paths to ensure that we only forward the essential ones. 

Let $k$ be fixed, and define for $v \in V$ and $\ell \in \{ 1, 2, \dotsc, k\}$ the set family
\[ \paths_{v,\ell} = \bigl\{ U \subseteq V \colon \text{there is an $\ell$-path with node set $U$ that ends at $v$} \bigr\}\,.\]
By definition, we have that 
\[ \paths_{v,\ell} = \bigcup_{u \in N(v)} \bigl\{ \{ v \} \cup U \colon U \in \paths_{u,\ell-1} \text{ and } v \notin U \bigr\}\,. \]
The basic idea of the algorithm is that instead of explicitly constructing the families $\paths_{v,1}, \paths_{v,2} \dotsc, \paths_{v,k}$, we iteratively construct families $\widehat{\paths}_{v,1}, \widehat{\paths}_{v,2} \dotsc, \widehat{\paths}_{v,k}$, where $\widehat{\paths}_{v,\ell}$ is a minimal $(k-\ell)$-representative family for $\paths_{v,\ell}$.

Specifically, the algorithm proceeds as follows:
\begin{enumerate}
    \item For $\ell = 1$, we have that $\paths_{v,1} = \{ \{ u,v \} \colon u \in N(v) \}$. Each node $v$ can obtain full information about $N(v)$ in single round, and then compute $\widehat{\paths}_{v,1}$ locally.
    \item For $\ell \ge 2$, after the families $\widehat{\paths}_{v,\ell-1}$ have been constructed, each node $v \in V$ broadcasts the family $\widehat{\paths}_{v,\ell-1}$ to all of its neighbours. Since $\widehat{\paths}_{v,\ell-1}$ is a minimal $(k-\ell+1)$-presentative family for $\paths_{v,\ell-1}$, and $\paths_{v,\ell-1}$ consist of sets of size $\ell$, we have that
    $|\widehat{\paths}_{v,\ell-1}| \le \binom{k+1}{\ell-1} \le 2^{k+1}$.
    In particular, $\widehat{\paths}_{v,\ell-1}$ can be encoded using $O\bigl(k \binom{k+1}{\ell-1} \log n\bigr)$ bits, and broadcasting it to all neighbours can be done in $O\bigl(k \binom{k+1}{\ell-1}\bigr)$ rounds.
    \item Each node $v \in V$, after having received $\widehat{\paths}_{u,\ell-1}$ for each $u \in N(v)$, locally constructs the set
    \[ \paths'_{v,\ell} = \bigcup_{u \in N(v)} \bigl\{ \{ v \} \cup U \colon U \in \widehat{\paths}_{u,\ell-1}  \text{ and } v \notin U  \bigr\}\,.\]
    We observe that $\paths'_{v,\ell}$ is $(k-\ell)$-representative for $\paths_{v,\ell}$. If $W \subseteq V$ is a set of size $k-\ell$ that does not intersect some $U \in \paths_{v,\ell}$, then there is some $u \in N(v)$ and $U' \in \paths_{u,\ell-1}$ such that $W \cup \{ v \}$ does not intersect $U'$. Since $\widehat{\paths}_{u,\ell-1}$ is $(k-\ell+1)$-representative for $\paths_{u,\ell-1}$, there is $R \in \widehat{\paths}_{u,\ell-1}$ such that $R$ does not intersect $W \cup \{ v \}$, and thus $R \cup \{ v \} \in \paths'_{v,\ell}$ does not intersect $W$. Thus, computing a minimal $(k-\ell)$-representative family $\widehat{\paths}_{v,\ell}$ for $\paths'_{v,\ell}$ gives a $(k-\ell)$-representative family for $\paths_{v,\ell}$ by \lemmaref{lemma:rep-transitivity}.
 \end{enumerate}
Clearly, there is a $k$-path terminating at $v$ if and only if $\widehat{\paths}_{v,k}$ is non-empty. There are $k$ phases corresponding to $\ell = 1, 2, \dotsc, k$ in the algorithm, and each of these phases runs in $O\bigl(k \binom{k+1}{\ell-1}\bigr)$ rounds, where the hidden constant does not depend on $\ell$. Since it holds that
is
\[ \sum_{\ell = 1}^k k \tbinom{k+1}{\ell-1} = k\sum_{\ell = 1}^k \tbinom{k+1}{\ell-1} \le k2^k\,,\]
the total running time is $O(k 2^k)$.

Finally, let us observe that it is easy to modify this algorithm to \emph{find} a $k$-path, in the sense that each node $v \in V$ that is an endpoint of a $k$-path has full knowledge of a single $k$-path: we annotate each set $U \in \widehat{\paths}_{v,\ell}$ with a `witness' $\ell$-path with node set $U$ terminating at $v$. This can be done without increasing the asymptotic communication cost.
 
\begin{theorem}
Finding $k$-paths can be done in $O(k2^k)$ rounds in the broadcast \congest{} model.
\end{theorem}
 
\paragraph{Finding cycles.} We first note that it is easy to adapt the $k$-path algorithm to detect $k$-cycles that contain a fixed node $w \in V$. We modify the definition of $\paths_{v,\ell}$ to require that the paths have $w$ as a starting point; otherwise the algorithm proceeds as before. If any neighbour of $w$ detects a $(k-1)$-path starting from $w$, then there is a $k$-cycle containing~$w$; likewise, if there is a $k$-cycle containing $w$, then some neighbour of $w$ will detect a $(k-1)$-path starting from $w$. Running this cycle detection algorithm for all choices of the starting node in parallel gives allows us to detect $k$-cycles in time $O(k2^k n)$:

\begin{theorem}
In the broadcast \congest{} model,
\begin{enumerate}
    \item finding $k$-cycles containing a fixed node $w \in V$ can be done in $O(k2^k)$ rounds, and
    \item finding $k$-cycles can be done in $O(k2^k n)$ rounds.
\end{enumerate} 
\end{theorem}

\paragraph{Finding trees.} 

We now show how to extend the path detection algorithm to any target graph $H$ that is acyclic. Let $H$ be a tree on $k$ nodes. Fix an arbitrary node of $H$ as a root, and identify the nodes of $H$ with $1,2,\dotsc, k$ so that if $i$ is a descendant of $j$ then $j > i$. In particular, node $k$ is the root. Denote by $H[i]$ the subtree of $H$ rooted at node $i$, and denote by $s_i$ the number of nodes in $H[i]$.
 
For $v \in V$ and $i \in \{ 1, 2, \dotsc, k\}$, we define the set family
\[ \trees_{v,i} = \bigl\{ U \subseteq V \colon \text{there is a copy of $H[i]$ with node set $U$ and root $v$} \bigr\}\,.\]
Again, if we know $\trees_{u,j}$ for all $u \in N(v)$ and and $j < i$, we can compute $\trees_{v,i}$. Let $c_1, c_2, \dotsc, c_\ell$ be the children of $i$ in $H$. Then $\trees_{v,i}$ contains exactly the sets of form 
\[ \{ v \} \cup U_1 \cup U_2 \cup \dotsb \cup U_\ell\,, \]
where $U_x \in \trees_{u,c_x}$ for some $u \in N(v)$ and $U_x \cap U_y = \emptyset$ for all $x \neq y$.

Similarly to the path algorithm, for $i = 1,2,\dotsc,k$, we iteratively construct $(k-s_i)$-representative family $\widehat{\trees}_{v,i}$ for $\trees_{v,i}$:
\begin{enumerate}
    \item If $i$ is a leaf in $H$, then $\widehat{\trees}_{v,i} = \trees_{v,i} = \{ \{ v \} \}$.
    \item If $i$ is an internal node in $H$ with children $c_1, c_2, \dotsc, c_\ell$, then we construct $\trees'_{v,i}$ by taking all sets of form
    \[ \{ v \} \cup U_1 \cup U_2 \cup \dotsb \cup U_\ell\,, \]
    where $U_x \in \widehat{\trees}_{u,c_x}$ for $u \in N(v)$ and $U_x \cap U_y = \emptyset$ for all $x \neq y$, and compute a minimal $(k-s_i)$-representative family $\widehat{\trees}_{v,i}$ for $\trees'_{v,i}$.
\end{enumerate}
We observe that after step $i$, each node $v$ knows if there is a copy of $H[i]$ that is rooted at $v$, which implies the correctness of the algorithm. Time complexity follows by the same arguments as in the path detection algorithm. To summarise:

\begin{theorem}
For any tree $H$ on $k$ nodes, $H$-subgraph detection can be done in $O(k2^k)$ rounds in the broadcast \congest{} model.
\end{theorem}

\paragraph{Finding pseudotrees.} Recall that a pseudotree is a connected graph that has at most one cycle, that is, a graph that consist of a tree and at most one additional edge. Let $H$ be a pseudotree on $k$ vertices, and let us assume that it has a cycle. Let $e = \{ u_1, u_2 \}$ be an arbitrary edge on the cycle.

We slightly modify our tree detection algorithm to find a copy of $H$ as follows. Consider the tree $H'$ obtained from $H$ by deleting the edge $e$. We identify the nodes of $H'$ with $1,2, \dotsc, k$ as before, using $u_2$ as the root node $k$; let us assume that $u_1$ is identified with $j$.

The basic idea is to modify the tree detection algorithm to keep track which nodes of the input graph can play the role of $j$ in a copy of $H'$.
Specifically, when constructing the sets $\widehat{\trees}_{v,i}$, we keep track of which node plays the role of $j$ in the copies of $H'[i]$ we have found so far. This requires the following changes to the tree detection algorithm:
\begin{itemize}
    \item For $i$ that is not on the unique path from $j$ to $k$, no changes are required.
    \item For the step corresponding to $j$, each node $v \in V$ performs the step as before, but marks itself as the node playing the role of $j$ when broadcasting the constructed representative family $\widehat{\trees}_{v,j}$.
    \item For $i$ that is on the path from $j$ to $k$, each node $v \in V$ keeps track of $\widehat{\trees}_{v,i}$ separately for each possible choice of $j$ it sees.
\end{itemize}
After the algorithm is finished, all nodes check if they see a copy of $H'$ where the node $j$ is one of their neighbours, which happens if and only if a copy of $H$ exists in the graph. At each step, each node has to broadcast at most $n$ copies of the representative family, giving us the following:

%
%

\begin{theorem}
For any pseudotree $H$ on $k$ nodes, $H$-subgraph detection can be done in $O(k2^kn)$ rounds in the broadcast \congest{} model.
\end{theorem}

\section{Enumerating cliques and short cycles in degenerate graphs\label{sec:degeneracy}}

\paragraph{Acyclic orientations with bounded outdegree.}

Our algorithms for enumerating cliques and cycles in sparse graphs exploit the fact that degenerate graphs have acyclic orientations with bounded outdegree. Once we have such an orientation, the nodes can co-ordinate how to distribute the information about the edges present in the input graph, which avoids having too much congestion over a single communication link.

For any $\alpha \in \N$, we say that $\sigma$ is an $\alpha$-bounded orientation~\cite{chrobak91orientations} of $G = (V,E)$ if
\begin{enumerate}[noitemsep]
    \item every node $v \in V$ has $\outdeg_\sigma(v) \le \alpha$, and
    \item the orientation $\sigma$ is acyclic. 
\end{enumerate}
For brevity, we call such orientations simply $\alpha$-orientations. The class of degenerate graphs can be characterised by the existence of such orientations.

\begin{lemma}\label{lemma:orientation}
 A graph $G$ is $d$-degenerate if and only if there exists a $d$-orientation of $G$.
\end{lemma}

Barenboim and Elkin~\cite{Barenboim2010} gave efficient distributed algorithms for computing such orientations under the name Nash--Williams forests-decompositions~\cite{nash-williams64decomposition}. For the sake of completeness, we formulate more or less the same algorithm here. Barenboim and Elkin~\cite{Barenboim2010} also show how to compute such orientations even without knowing either the degeneracy (equivalently arboricity), or alternatively, the number of nodes in the network in $O(\log n)$ rounds using messages of size $O(\log n)$.

\begin{lemma}
    If $G$ is $d$-degenerate, then it has at most $nd$ edges.
    \label{lemma:degenerate-edge-count}
\end{lemma}
\begin{proof}
Consider the $d$-orientation given by \lemmaref{lemma:orientation}. Counting the outgoing edges yields
\[
    |E| = \sum_{v \in V} \outdeg(v) \le nd. \qedhere
\]
\end{proof}

Let $G$ be a $d$-degenerate graph and $C > 2$ be a constant. For all $i \ge 0$, define iteratively the node sets $V_0, V_1, V_2, \dotsc$ as follows:
\begin{align*}
    V_0 &= V\,, \\
    V_{i+1} &= \{ v \in V_i : \deg(v) > Cd \}\,.
\end{align*}
We note that each node can compute $i_v = \max \{ i \in \N \colon v \in V_i \}$ in $i_v$ rounds.

\begin{lemma}
 We have that $|V_{i+1}| \le 2/C \cdot |V_i|$.
\end{lemma}
\begin{proof}
    Suppose the opposite holds, that is, there are $h > 2/C \cdot |V_i|$ nodes $v \in V_i$ with $\deg(v) > Cd$. As $G[V_i] \subseteq G$ is $d$-degenerate, counting the number of edges incident to these $h$ nodes yields that there are at least  
\[
    \frac{Cdh}{2} > \frac{2 \cdot |V_i| Cd}{2C} = |V_i| \cdot d
\]
edges, which contradicts \lemmaref{lemma:degenerate-edge-count}.
\end{proof}

By the above lemma, we get that every iteration we lose a constant fraction $1/c$ of nodes, where $c = C/2$. Thus, $|V_i| \le n/c^i$ and after $r = \log n / \log c + 1 = \log_c n$ iterations we have 
\[
 |V_r| = \frac{n}{c^{\log_c n + 1}} \le 1.
\]

We can now compute an acyclic orientation of $G$ using the above scheme. Each node $v \in V_{i} \setminus V_{i+1}$ removed in iteration $i$ orients any edge $\{v,u\}$ with $u \in V_{i+1}$ towards $u$. Clearly there will be at most $Cd$ edges pointing away from $v$. 

\begin{lemma}[\cite{Barenboim2010}]\label{lemma:computed-orientation}
    Let $G$ be a $d$-degenerate graph. For any constant $\varepsilon > 0$, we can compute a $(2d+\varepsilon)$-orientation of $G$ in $O(\log n)$ rounds in the broadcast \congest{} model.
\end{lemma}

\paragraph{Enumerating cliques.}

We start with a $k$-clique enumeration algorithm that works in $d$-degenerate graphs for any $k \le d$. Note that the scenario $k > d$ is fatuous, as a $d$-degenerate graph cannot contain a clique with more than $d$ nodes. Let $G = (V,E)$ be the input graph. The algorithm is as follows:
\begin{enumerate}
    \item Compute the $\alpha$-orientation of $G$ for $\alpha \in \Theta(d)$.
    \item Each $v \in V$ broadcasts the endpoints $\nout(v)$ of outgoing edges to all its neighbours. 
    \item Each $v \in V$ locally constructs the induced subgraph $G[F]$, where 
        \[
         F = \{ \{u,w\} : u \in N^+(v), w \in \nout(u) \}.
        \]
    \item Each $v \in V$ outputs all $k$-cliques in $G[F]$.
\end{enumerate}
\figureref{fig:clique-example} illustrates the information gathered by the above algorithm. 

In order to analyse the algorithm, first observe that clearly the algorithm clearly only outputs $k$-clique that are present in $G$. We now argue that if $G$ contains a $k$-clique, then some node will list it in the output. Let $K \subseteq V$ be a $k$-clique in $G$. Note that $G[K]$ must contain a sink node $v \in K$ with respect to the $\alpha$-orientation $\sigma$: if such a sink node does not exist, then every node in $G[K]$ would have positive outdegree implying that $\sigma$ is not acyclic, which is absurd. Hence, the sink node $v \in K$ will receive the set $\nout(u)$ from every $u \in K$. In particular, we have 
\[
    K \subseteq \bigcup_{u \in N(v)} \nout(u).
\]
    Therefore, the sink node $v$ detects the clique $K$ in the subgraph $H$ it constructs in the third step, and thus, will list it in its output.

It remains to analyse the time complexity of the algorithm. By \lemmaref{lemma:computed-orientation}, the first step takes $O(\log n)$ rounds. The second step takes $O(\alpha)$ rounds, as broadcasting a single identifier requires $O(\log n)$ bits to be transmitted along each edge. Since there are at most $\alpha$ such identifiers to be broadcast (as there are at most $\alpha$ outgoing edges) and $\alpha \in O(d)$, we get that this step takes $O(d)$ time. The remaining steps require no communication, and hence, all nodes declare their output within $O(\log n + d)$ rounds. 

\begin{figure}
\begin{center}
 \includegraphics[page=6]{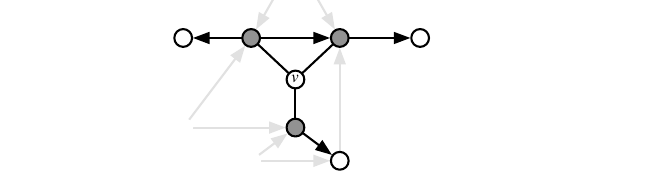}
\end{center}
    \caption{\label{fig:clique-example}Example of the information gathered by the clique detection algorithm. Once an $\alpha$-orientation of the graph has been computed, the algorithm ensures that node $v$ obtains all outgoing edges of its neighbours (black edges). Neighbours of $v$ are are dark gray and the edges not observed by $v$ are gray. Here, node $v$ detects the triangle its part of.}
\end{figure}

\begin{theorem}
    Let $k \le d$. In any $d$-degenerate graph $G$, enumerating $k$-cliques takes $O(d + \log n)$ rounds in the broadcast \congest{} model. 
\end{theorem}

As triangles are 3-cliques, we immediately get the following corollary.

\begin{corollary}
    In any $d$-degenerate graph $G$, enumerating triangles takes $O(d + \log n)$ rounds in the broadcast \congest{} model. 
\end{corollary}

\paragraph{Enumerating 4-cycles.} In the case of 4-cycles, it turns out that we can use the same algorithm as for enumerating cliques, except we check for the existence of a 4-cycle in the locally constructed subgraph $G[F]$ instead of a $k$-clique. Let $C = \{u,v,w,z\}$ be a 4-cycle in $G$. As before, we first obtain $\alpha$-orientation $\sigma$ of the $d$-degenerate input graph $G$ using \lemmaref{lemma:computed-orientation}. Since the orientation is acyclic, the cycle graph $G[C]$ must have at least one source and a sink in the orientation. In particular, there are three possible cases (up to isomorphisms):
\begin{center}
 \includegraphics[page=2,scale=0.9]{figures.pdf}
\end{center}
To show that some node detects the cycle $C$, we consider each of these three cases. In each case, observe that $u$ sends $v$ the set $\nout(u)$ and $w$ sends the set $\nout(w)$ to $v$. Thus, node $v$ learns about the edges $\{u,z\}$ and $\{w,z\}$ as $z \in \nout(u) \cap \nout(w)$. Since $v$ trivially knows about the edges $\{u,v\}$ and $\{v,w\}$, we get that in each case node $v$ learns about all edges in $C$ and lists $C$ in its output. Note that if we want that each $4$-cycle is listed by exactly one node, the nodes can detect if multiple nodes would be listing the same cycle $C$ and output $C$ only if they are the node with smallest identifier having knowledge about it.

As the algorithm communicates the same information as the clique detection algorithm given before, we get that the time complexity of detecting 4-cycles is the same. 

\begin{theorem}
    In any $d$-degenerate graph $G$, enumerating $4$-cycles takes $O(d + \log n)$ rounds in the broadcast \congest{} model.
\end{theorem}

\paragraph{Finding 5-cycles.} In order to enumerate 5-cycles, we use the same underlying idea, but now nodes also communicate directed outgoing paths of length 2 instead of just single outgoing edges. With this modification in mind, the algorithm is as follows:
\begin{enumerate}
    \item Compute the $\alpha$-orientation of $G$ for $\alpha \in \Theta(d)$.
    \item Each $v \in G$ broadcasts its endpoints $\nout(v)$ of outgoing edges to all its neighbours. 
    \item Each $v \in G$ broadcasts the set of outgoing length-2 paths
        \[  
            L(v) = \{ \{u,w \} : u \in \nout(v), w \in \nout(u) \}
        \]
        to all its neighbours.
    \item Each $v \in G$ locally constructs the subgraph $G[F]$, where 
        \[
         F = \left\{ \{u,w\} : u \in N^+(v), w \in \nout(u) \right \} \cup \bigcup_{u \in \nout(v)} L(u).
        \]
    \item Each $v \in V$ outputs all $5$-cycles in $G[F]$.
\end{enumerate}

Let $C = \{u,v,w,x,y\}$ be a 5-cycle in $G$. To see that the algorithm outputs the cycle~$C$, it suffices to show that some node $v \in C$ learns about all the edges in $C$. As the $\alpha$-orientation is acyclic, the cycle $C$ can be oriented in the following three ways (up to isomorphisms):
\begin{center}
 \includegraphics[page=4,scale=0.9]{figures.pdf}
\end{center}
To see that this is the case, one can readily observe that the length of the longest directed path in an acyclically oriented $5$-cycle has to be either $2$, $3$ or $4$, and the length of this path determines the orientation of all other edges.

We argue that node $v$ detects the cycle $C$ in all cases. Observe that node $x$ broadcasts $\nout(x)$ to $u$ in the second step and node $u$ broadcasts $\nout(u)$ and $L(u)$ to $v$ in the second and third steps, respectively. Since $y \in \nout(x)$, we have $\{x,y\} \in L(u)$ and $x \in \nout(u)$. Therefore, $v$ learns about the path $(u,x,y)$. Since $w$ broadcasts $\nout(w)$ to $v$, node $v$ also learns about the edge $\{w,y\}$. Since node $v$ can trivially detect the edges $\{u,v\}$ and $\{v,w\}$, it follows that node $v$ detects the cycle $C$ in the final step.

The time complexity of the first step is again $O(\log n)$. The second step consists of broadcasting the set of up to $\alpha$ identifiers, which takes $O(\alpha)$ rounds as before. In the third step, each node $v \in G$ broadcasts $L(v)$ which may contain up to $O(\alpha^2)$ edges. Thus, the third step takes $O(\alpha^2)$ rounds. No communication occurs in the final steps of the algorithm, which yields that the total time complexity of the algorithm is $O(\log n + \alpha^2)$. As we can choose $\alpha \in \Theta(d)$, we obtain the following result.

\begin{theorem}
    In any $d$-degenerate graph $G$, enumerating $5$-cycles takes $O(d^2 + \log n)$ rounds in the broadcast \congest{} model.
\end{theorem}

\begin{figure}
\begin{center}
 \includegraphics[page=5]{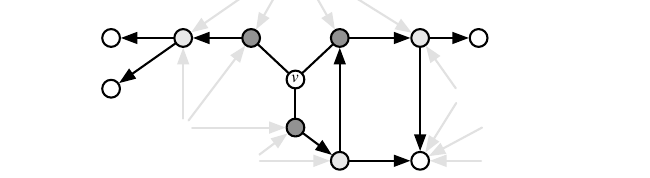}
\end{center}
    \caption{Example of information the nodes gather when executing the 5-cycle detection algorithm. Once an $\alpha$-orientation of the graph has been computed, the algorithm ensures that node $v$ obtains all outgoing length-1 and length-2 paths (black directed edges) of its neighbours. Neighbours of $v$ are are dark gray, nodes at distance two are light gray and nodes at distance three are white.}
\end{figure}

\section{Lower bounds}

\paragraph{Lower bounds for long cycles.}

We first prove our lower bounds for the detection of long cycles:  finding $k$-cycles for even $k$ requires $\Omega(n^{1/2}/\log n)$ rounds, and there is no algorithm for finding $k$-cycles for $k \ge 6$ that runs in $O(f(d) n^{\delta})$ rounds on $d$-degenerate graphs for any function $f$ and $\delta < 1/2$.
We use a slight modification of the cycle detection lower bound construction of Drucker et al.~\cite{drucker13}; that is, we prove the lower bound by a reduction to known set disjointness lower bounds in two-player communication complexity.

Specifically, given sets $A$ and $B$ over $M$-element universe, we construct a $n$-node graph $G_{A,B} = (V_A \cup V_B, E)$ such that (1) $G[V_A]$ depends only on $A$, (2) $G[V_B]$ depends only on $B$, (3) there are at most $C$ edges crossing the cut $(V_A,V_B)$, and (4) $G_{A,B}$ contains a cycle of length $k$ if and only if $A$ and $B$ are non-disjoint. This allows us to use any $k$-cycle detection to solve 2-party set disjointness; Alice simulates all nodes in $V_A$, Bob simulates nodes in $V_B$, and messages over edges crossing the cut $(V_A, V_B)$ are sent between players. By known lower bounds,  $\Omega(M)$ bits have to be sent over the cut, giving lower bound of $\Omega(M/C\log n)$ rounds in the \congest{} model.

We show that for any $k \ge 6$ and all $N \in \N$, there is a construction satisfying the above with parameters $M \in \Theta(N^2)$, $n \in \Theta(N^2)$ and  $C = 2N$, with the additional property that any graph obtained from this construction is $2$-degenerate. This implies both of the lower bounds mentioned above.

\begin{figure}
\begin{center}
 \includegraphics[page=8,scale=1.1]{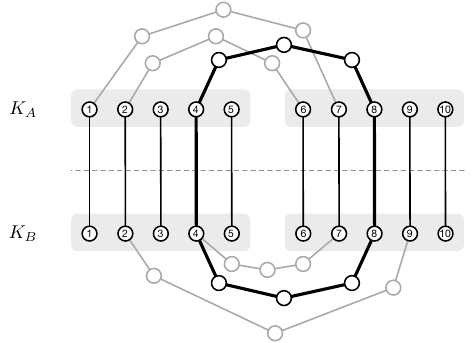}
\end{center}
    \caption{Example of the lower bound construction for $k = 8$, $N = 5$ and $M = 25$. A $8$-cycle is highlighted in black; note that a 8-cycle can occur only when the corresponding paths are present both in $K_A$ and $K_B$.}
\end{figure}

Let $k \ge 6$ and $N \in \N$ be fixed, and let $A, B \subseteq 2^{[N^2]}$ be a set disjointness instance. Let $\ell_1 = \lfloor k/2 \rfloor$ and  $\ell_2 = \lceil k/2 \rceil$. We take two disjoint copies $K_A$ and $K_B$ of a complete bipartite graph $K_{N,N}$ on $2N$ nodes; assume that the nodes of the both copies are labelled in a consistent manner with integers $1,2,\dotsc, 2N$, and that the edges are likewise consistently labelled with the elements of $2^{[N^2]}$. We now finish the construction as follows:
\begin{enumerate}
    \item We connect nodes labelled with the same integer in $K_A$ and $K_B$ with an edge.
    \item We remove each edge in $K_A$ if the corresponding element in $2^{[N^2]}$ is not in $A$, and we remove each edge in $K_B$ if the corresponding element in $2^{[N^2]}$ is not in $B$.
    \item We replace all remaining edges in $K_A$ with $(\ell_1 - 1)$-paths, and all remaining edges in $K_B$ with $(\ell_2 - 1)$-paths.
\end{enumerate}
It is now easy to verify that there are no $k$-cycles remaining inside either $K_A$ or $K_B$, and there is a $k$-cycle if and only if $A \cap B \neq \emptyset$. The final graph has at most $(k-4)N^2 + 4N$ nodes and at most $(k-2) N^2 + 2N$ edges; the cut between $K_A$ and $K_B$ has size $2N$. Finally, it is easy to see that the graph has degeneracy $2$. For each path added in step (3), we pick an internal node and orient all edges away from it, and orient the cut edges between $K_A$ and $K_B$ arbitrarily; the resulting orientation has maximum outdegree two.

\paragraph{Lower bounds for short cycles.} We observe that the lower bounds given by Drucker et al.~\cite{drucker13} imply that finding $4$-cycles and $5$-cycles in $d$-degenerate graphs requires $\Omega(d / \log n)$ rounds: a careful examination shows that their lower bound graph for $4$-cycles has degeneracy $\Theta(n^{1/2})$ and the lower bound graph for $5$-cycles has degeneracy $\Theta(n)$.

\section{Extensions to the supported \texorpdfstring{\boldmath\congest{}}{CONGEST}}

\paragraph{The supported \texorpdfstring{\boldmath\congest{}}{CONGEST} model.}

Finally, we discuss how our results can be extended to the supported \congest{} proposed by Schmid and Suomela~\cite{schmid13local-sdn} for studying distributed algorithms in the context of \emph{software defined networking} (SDN). 

In the supported \congest{} model, the underlying physical network topology is represented by a \emph{support graph} $G = (V,E)$ and the actual \emph{input graph} $H$ for the distributed algorithm is an arbitrary subgraph of the support. The input graph $H = (U,F)$ inherits the identifiers of the support graph and can be seen as the current logical state of the network. Each node $v$ in the network $G$ is aware of the full topological information about $G$, but node $v$ has only local information about the current logical state of the network, that is, the actual input graph $H$.
That is, initially node $v$ will know only which of the edges incident to it in $G$ are also present in $H$. Note that supported \congest{} can be seen as a relaxation of the congested clique model: in the case of the congested clique model, the support graph is a clique. 

While many symmetry breaking problems are trivial in the supported \congest{} model, the model remains interesting from the perspective of subgraph detection. Even though all nodes will a priori know that any edge $\{u,v\}$ not present in $G$ will also not be in $H$, only nodes $u$ and $v$ initially know whether $\{u,v\} \in F$ as well. Thus, the difficulty now becomes verifying which of the possible subgraphs in $G$ remain in $H$.

\paragraph{Subgraph enumeration in sparse support graphs.} We can modify the algorithms given in \sectionref{sec:degeneracy} to run faster in the supported \congest{} model, in the case where the support graph has bounded degeneracy. First, all nodes can locally determine the degeneracy of the support graph $G$ and compute (the same) $d$-orientation $\sigma$ of $G$ without any communication. The orientation $\sigma$ restricted to the input graph $H$ will also be an $d$-orientation of the input $H$. This saves the additive $O(\log n)$ term in the running time. Moreover, when the nodes communicate their outgoing edges, as every node $v \in V$ knows which outgoing edges of $\sigma$ \emph{can be} present incident to $u \in V$, it suffices that each node $u$ broadcasts $d$-length bit string encoding which edges of $G$ are present in the input $H$. With these modifications in mind, we get the following result. 

\begin{theorem}
In any $d$-degenerate support graph $G$  in the supported \congest{} model, 
   \begin{enumerate}
       \item enumerating $k$-clique can be done in $O(d/\log n)$ rounds for any $k$,
       \item enumerating $4$-cycles can be done in $O(d/\log n)$ rounds, and 
       \item enumerating $5$-cycles can be done in $O(d^2/\log n)$ rounds.
   \end{enumerate} 
\end{theorem}
In particular, in the case that the support graph has degeneracy $O(\log n)$, we can enumerate cliques and 4-cycles in constant number of rounds in the supported \congest{} model. 

\paragraph{Cycle detection lower bounds.} Finally, we note that all lower bounds from the present work as well as those by Drucker et al.~\cite{drucker13} hold in the supported \congest{}. All these lower bounds are obtained by starting from a fixed base graph $G$, and removing certain edges to obtain a graph $G'$ that encodes a set disjointness instance. Taking the base graph $G$ as the support graph and $G'$ as the input graph yields identical lower bounds for the supported \congest{}.

\section{Conclusions}

\paragraph{Subgraph detection.} We note that there still remain major open questions regarding subgraph detection in the \congest{} model:
\begin{itemize}
    \item What is the complexity of general subgraph detection? In particular, can the subgraph detection problem solved in linear number of rounds for any constant-size target graph $H$, or are the target graphs that require a superlinear number of rounds? Note that $k$-cliques, which appear to be the most difficult case for centralised subgraph detection, can be trivially detected in $O(n)$ rounds for any $k$.
    \item What is the precise complexity of detecting even-length cycles? Is there an algorithm matching the $\tilde{\Omega}(n^{1/2})$ lower bound we give, or can this lower bound be improved?
\end{itemize}

\paragraph{Fixed-parameter techniques.} We remark that the many algorithmic techniques from fixed-parameter cycle and path detection algorithms seem to translate very naturally to the distributed limited bandwidth setting, as they effectively either (a) partition the task at hand into multiple independent instances of an easier task (e.g.\ colour-coding), or (b) compress the intermediate results of the computation (e.g.\ representative families). As noted before, the seminal colour-coding technique of Alon et al.~\cite{alon1995color} has been used by various authors for distributed algorithms~\cite{censor2015algebraic,2017arXiv170504898E,firscher2017}; indeed, the results in the present work can also be derived via colour-coding, albeit with slightly worse dependence on $k$, by an easy modification of the congested clique cycle detection algorithm of Censor-Hillel et al.~\cite{censor2015algebraic}. We also expect that the $O(k2^k)$ round running time for path detection could be improved with randomisation, for example by using algebraic \emph{sieving} techniques~\cite{sieves2017}.

More generally, such fixed-parameter techniques are applicable in the centralised setting beyond subgraph detection, and we expect this to be the case also in the distributed setting. For example, it seems likely that constant-round distributed algorithms for the \emph{graph motif} problem can be obtained either via colour-coding or algebraic sieving~\cite{sieves2017,betzler2008parameterized}. Indeed, this general line of research may even have practical relevance, as evidenced by the efficient parallel implementation of fixed-parameter graph motif algorithms by Bj\"{o}rklund et al.~\cite{doi:10.1137/1.9781611973754.10}.

\medskip

\paragraph{Acknowledgements} This work was supported in part by the Academy of Finland, Grants 285721 and 1273253. We thank Juho Hirvonen, Dennis Olivetti, Rotem Oshman, Chris Purcell, Stefan Schmid and Jukka Suomela for valuable comments and discussions. 

\DeclareUrlCommand{\Doi}{\urlstyle{same}}
\renewcommand{\doi}[1]{\href{http://dx.doi.org/#1}{\footnotesize\sf doi:\Doi{#1}}}

\bibliographystyle{plainnat}
\bibliography{congest-subgraphs}

\end{document}